\documentclass[11pt, a4paper]{tboege-preprint}
\pdfoutput=1

\makeatletter
\@namedef{subjclassname@2020}{%
  \textup{2020} Mathematics Subject Classification}
\makeatother

\usepackage{algorithm}
\usepackage{algpseudocode}
\usepackage{tboege-bordermatrix}

\newcommand{\sa}{\SF{sa}}

\newcommand{\A}{\CC{A}}
\newcommand{\Id}{\BBm1}
\newcommand{\PD}{\RM{PD}}

\title{Selfadhesivity in Gaussian \\ conditional independence structures}

\author{Tobias Boege}
\address{\mbox{Tobias Boege, Department of Mathematics, KTH Stockholm, Sweden}}
\email{post@taboege.de}
\date{\today}


\begin{document}

\begin{abstract}
Selfadhesivity is a property of entropic polymatroids which guarantees
that the polymatroid can be glued to an identical copy of itself along
arbitrary restrictions such that the two pieces are independent given
the common restriction.
We show that positive~definite matrices satisfy this condition as well
and examine consequences for Gaussian conditional independence structures.
New axioms of Gaussian CI are obtained by applying selfadhesivity to
the previously known axioms of structural semigraphoids and orientable
gaussoids.
\end{abstract}

\maketitle

\section{Introduction}

In matroid theory, the term \emph{amalgam} refers to a matroid in which
two smaller matroids are glued together along a common restriction, similar
to how four triangles can be glued together along edges to form the boundary
of a tetrahedron.
This concept is meaningful for conditional independence (CI) structures
as~well. The bridge from the geometric (matroid-theoretical) concept to
probability theory (conditional independence) was built by Matúš \cite{MatusAdhesive}
who defined a special kind of amalgam, the \emph{adhesive extension},
for polymatroids and proved that such extensions always exist for
entropic polymatroids with a common restriction.

The purpose of this article is two-fold: First, it is to extend
this methodology beyond polymatroids and to introduce a derived collection
of amalgamation properties known as \emph{selfadhesivity} for general conditional
independence structures.
Second, this general treatment of selfadhesivity is driven by its
applications to Gaussian instead of discrete CI~inference. The main
result, \Cref{thm:Adhe}, shows that, also in the Gaussian setting, adhesive
extensions (of covariance matrices) exist and are even unique. We use the
non-trivial structural constraints implied by this result to derive new
axioms for Gaussian conditional independence structures.
These results are heavily based on computations. All source code and
further details on computations are provided on the mathematical research
data repository \emph{MathRepo} hosted by the Max-Planck Institute for
Mathematics in the Sciences; the output data, due to its size, is
available only on the archiving service \emph{KEEPER} of the Max-Planck
Society:
\begin{center}
  \begin{tabular}{rl}
  MathRepo: & \url{https://mathrepo.mis.mpg.de/SelfadhesiveGaussianCI/} \\
  KEEPER:   & \url{https://keeper.mpdl.mpg.de/d/fbfe463162e94a14ac28/}
  \end{tabular}
\end{center}

\section{Preliminaries}

\paragraph{Gaussian conditional independence}

Let $N$ be a finite ground set indexing jointly distributed random variables
$\xi = (\xi_i : i \in N)$. By convention, elements of $N$ are denoted by
$i,j,k, \dots$ and subsets by $I,J,K, \dots$. Elements are identified with
singleton subsets of $N$ and juxtaposition of subsets abbreviates set union.
Thus, an expression such as $iK$ is shorthand for $\Set{i} \cup K$ as a
subset of~$N$. The complement of $K \subseteq N$ is $K^\co$. The set of all
$k$-element subsets of $N$ is $\binom{N}{k}$ and the powerset of $N$ is $2^N$.

We are mostly interested in Gaussian (i.e., multivariate normal) distributions.
These distributions are specified by a small number of parameters, namely
by the mean vector $\mu \in \BB R^N$ and the covariance matrix $\Sigma \in \PD_N$,
where $\PD_N$ is the set of positive definite matrices.
Throughout this article, ``Gaussian'' means ``regular Gaussian'', i.e., the
covariance matrix is strictly positive~definite. For positive semidefinite
covariance matrices, which lie on the boundary of $\PD_N$, the CI~theory
is algebraically more complicated and valid inference properties for regular
Gaussians can fail to be valid for singular ones; see \cite[Section~2.3.6]{Studeny}.

The following result summarizes basic facts from algebraic statistics relating
subvectors of $\xi$ and their (positive~definite) covariance matrices.
It can be found, for instance, in \S2.4 of \cite{Sullivant}.
For $\Sigma \in \PD_N$ and $I,J,K\subseteq N$, let $\Sigma_{I,J}$ denote
the submatrix with rows indexed by $I$ and columns by~$J$. Submatrices of
the form $\Sigma_K \defas \Sigma_{K,K}$ are \emph{principal}. Dual to a
principal submatrix is its \emph{Schur complement} $\Sigma^K \defas
\Sigma_{K^\co} - \Sigma_{K^\co,K} \Sigma_K^{-1} \Sigma_{K,K^\co}$
in~$\Sigma$. Its rows and columns are indexed by $K^\co$ and
its entries are functions of all the entries of $\Sigma$.
Principal submatrices of and Schur complements in positive~definite
matrices are also~positive~definite. The Schur complement construction
is valid in greater generality which we will need below as well. Let $A$
be any (not necessarily positive definite, or even square) matrix whose
rows are indexed by $IK$ and columns by $JK$, where $I,J,K$ are pairwise
disjoint, and suppose that the principal submatrix $A_K$ is invertible.
Then the Schur complement $A^K = A_{I,J} - A_{I,K} A_K^{-1} A_{K,J}$ is
well-defined and its rows are indexed by $I$ and its columns by~$J$.
See \cite{Zhang} for an introduction to theory of Schur complements
in matrix analysis.

\begin{theorem*}
  Let $\xi$ be distributed according to the (regular) Gaussian
  distribution with mean $\mu \in \BB R^N$ and covariance
  $\Sigma \in \PD_N$.  Let $K\subseteq N$.
\begin{itemize}[itemsep=0.6em]
\item The marginal vector $\xi_K = (\xi_k : k \in K)$ is a regular
  Gaussian in $\BB R^K$ with mean vector $\mu_K$ and covariance~$\Sigma_K$.
\item Given $y \in \BB R^K$, the conditional $\xi_{K^\co} \mid \xi_K = y$
  is a regular Gaussian in $\BB R^{K^\co}$ with mean vector $\mu_{K^\co} +
  \Sigma_{K^\co,K} \Sigma_K^{-1} (y - \mu_K)$ and covariance $\Sigma^K$.
\item Let a Gaussian distribution over $N = IJ$ be given with
  covariance~$\Sigma \in \PD_{IJ}$.  Then~the marginal independence
  $\CI{\xi_I,\xi_J}$ holds if and only if $\Sigma_{I,J} = 0$.
  \qedhere
\end{itemize}
\end{theorem*}

The general CI~statement $\CI{\xi_I,\xi_J|\xi_K}$, with $I,J,K$ pairwise
disjoint, is the result of marginalizing $\xi$ to $IJK$, conditioning on
$K$ and then checking for independence of $I$ and $J$. The previous
theorem implies the following algebraic CI~criteria for regular Gaussians:
\begin{align*}
  \CI{\xi_I,\xi_J|\xi_K}
    &\;\Leftrightarrow\; \left( \Sigma_{IJ} - \Sigma_{IJ,K}
       \Sigma_K^{-1} \Sigma_{K,IJ} \right)_{I,J} = 0 \\
    &\;\Leftrightarrow\; \Sigma_{I,J} - \Sigma_{I,K}
       \Sigma_K^{-1} \Sigma_{K,J} = 0              \tag{$\CIperp_1$} \label{eq:CIL} \\
    &\;\Leftrightarrow\; \Rk \Sigma_{IK,JK} = |K|. \tag{$\CIperp_2$}  \label{eq:CI}
\end{align*}
Here, $\Rk$ denotes the rank of a matrix and the last equivalence
follows from rank additivity of the Schur complement (see \cite{Zhang}).
Indeed, the matrix in \eqref{eq:CIL} is the Schur complement of $K$ in
$\Sigma_{IK,JK}$ and must have rank zero since the principal submatrix
$\Sigma_K$ has full rank $|K|$ already because it is positive~definite.
In particular, the truth of a conditional independence statement does not
depend on the conditioning event and it does not depend on the mean~$\mu$.
Hence, for CI purposes in this article, we identify regular Gaussians with
their covariance matrices~$\Sigma \in \PD_N$.

Rank additivity of the Schur complement also shows that the ``$\ge$'' part
of the rank condition in \eqref{eq:CI} always holds. Hence, the minimal rank
$|K|$ is attained if and only if all minors of $\Sigma_{IK,JK}$ of size
$|K|+1$ vanish.
But only a subset of these minors is necessary: by \eqref{eq:CIL} the rank
of $\Sigma_{IK,JK}$ is $|K|$ if and only if
$\Sigma_{I,J} = \Sigma_{I,K}\Sigma_K^{-1}\Sigma_{K,J}$ holds.
This is one polynomial condition for each $i \in I$ and $j \in J$, namely
$\det \Sigma_{iK,jK} = 0$ --- again by Schur complement expansion of the
determinant. These minors correspond to CI~statements of the form
$\CI{\xi_i,\xi_j|\xi_K}$. This proves the following ``localization rule''
for Gaussian conditional independence:
\[
  \label{eq:L} \tag{L}
  \CI{\xi_I,\xi_J|\xi_K} \;\Leftrightarrow\; \bigwedge_{i \in I, j \in J}
    \CI{\xi_i,\xi_j|\xi_K}.
\]
Rules of this form go back to \cite{MatusAscending}. A weaker localization
rule~\eqref{eq:L'} (discussed below) holds for all semigraphoids,
whereas the one presented above can be proved for compositional graphoids;
see~\cite{UnifyingMarkov} in the context of graphical models.
In~both cases, a general CI~statement is reduced to a conjunction of
\emph{elementary CI~statements} $\CI{\xi_i,\xi_j|\xi_K}$ about the
independence of two singletons.
We adopt the form $\CI{I,J|K}$ for CI~statements $\CI{\xi_I,\xi_J|\xi_K}$
without the mention of a random vector. These symbols are treated as
combinatorial objects and $\A_N \defas \{\, \CI{i,j|K} : ij \in
\binom{N}{2}, \allowbreak {K \subseteq N \setminus ij} \,\}$ is the
set of all elementary CI~statements.
The \emph{CI~structure} of $\Sigma$ is the set
\[
  \CIS{\Sigma} \defas \Set{ \CI{i,j|K} \in \A_N : \det \Sigma_{iK,jK} = 0 }.
\]
The localization rule shows that $\CIS{\Sigma}$ encodes the entire set
of true CI~statements for a Gaussian with covariance matrix~$\Sigma$
and with slight abuse of notation we employ statements such as
$\CI{I,J|K} \in \CIS{\Sigma}$.

It is important to note in this context that we treat only \emph{pure}
CI~statements, i.e., $\CI{I,J|K}$ where $I,J,K$ are pairwise
disjoint. Any general CI~statement with overlaps between the three sets
decomposes, analogously to the localization rule, into a conjunction of
pure CI~statements and functional dependence statements. For~a~regular
Gaussian, functional dependences are always false, so this is no
restriction in generality. In particular, the general statement $\CI{N,M|L}$,
which frequently appears later, is equivalent to $\CI{(N\setminus L), (M\setminus L)|L}$
which is pure provided that~${L \supseteq N \cap M}$.

\paragraph{Polymatroids and selfadhesivity}

A \emph{polymatroid} over the finite ground set $N$ is a function
$h: 2^N \to \BB R$ assigning to every subset $K \subseteq N$ a real
number, such that $h$ is
\begin{description}[itemsep=0pt]
\item[normalized] $h(\emptyset) = 0$,
\item[isotone]    $h(I) \le h(J)$ for $I \subseteq J$,
\item[submodular] $h(I) + h(J) \ge h(I \cup J) + h(I \cap J)$.
\end{description}
With the linear functional $\CId{I,J|K} \cdot h \defas h(IK) + h(JK) -
h(IJK) - h(K)$, submodularity can be restated as $\CId{I,J|K} \cdot h \ge 0$
for all pairwise disjoint $I, J, K$. If $h_\xi$ is the entropy
vector of a discrete random vector~$\xi$, i.e., $h_\xi(K)$ is the Shannon
entropy of the marginal vector $\xi_K$, then it is a polymatroid and
the quantity $\CId{I,J|K} \cdot h_\xi$ is known as the \emph{conditional
mutual information} $I(\xi_I; \xi_J | \xi_K)$. Its vanishing is equivalent
to the conditional independence $\CI{\xi_I,\xi_J|\xi_K}$.
Hence we may define the CI~structure of a polymatroid as
$\CIS{h} \defas \Set{ \CI{i,j|K} \in \A_N : \CId{ij|K} \cdot h = 0 }$.
These structures are called \emph{(elementary) semimatroids} in
\cite{MatusMatroids} and (equivalently, but based on properties of
multiinformation instead of entropy vectors) \emph{structural semigraphoids}
in \cite{StudenyStructural}.
Again, per \cite{MatusMatroids} a localization rule holds for them
which we use to interpret the containment of non-elementary CI~statements:
\[
  \label{eq:L'} \tag{$\text{L}'$}
  \CI{I,J|K} \in \CIS{h} \;\Leftrightarrow\; \bigwedge_{\substack{i \in I, j \in J, \\
    K \subseteq L \subseteq IJK \setminus ij}} \CI{i,j|L} \in \CIS{h}.
\]
This rule can be proved from the semigraphoid axioms and hence it holds
true also for Gaussians. In this case, it is equivalent to the shorter
rule~\eqref{eq:L} using that Gaussians are compositional graphoids.

Matúš in~\cite{MatusAdhesive} introduced the notions of adhesive extensions
and selfadhesive polymatroids to mimic a curious amalgamation property of
entropy vectors. The underlying construction is the \emph{Copy lemma}
of \cite{ZhangYeung}, also known as the \emph{conditional product};
see \cite[Section~2.3.3]{Studeny}.
For any polymatroid $g: 2^N \to \BB R$ and subset $L \subseteq N$
the \emph{restriction} $g|_L: 2^L \to \BB R$ given by $g|_L(K) \defas g(K)$,
$K \subseteq L$, is again a polymatroid.
Let $g$ and $h$ be two polymatroids on ground sets $N$~and~$M$, respectively,
and suppose that their restrictions $g|_L$ and $h|_L$ to $L = N \cap M$
coincide. A~polymatroid $f$ on $NM$ is an \emph{adhesive extension} of $g$
and $h$ if:
\begin{itemize}[itemsep=0pt]
\item $f|_N = g$ and $f|_M = h$,
\item $\CI{N,M|L} \in \CIS{f}$.
\end{itemize}
Since $L \subseteq N$ and $L \subseteq M$, the statement $\CI{N,M|L}$
is naturally equivalent to the pure CI~statement $\CI{N',M'|L}$ with
$N' = N \setminus L$ and $M' = M \setminus L$. In polymatroidal terms,
$N$~and $M$ are said to form a \emph{modular pair} in $f$ if this
CI~statement holds.

Next, suppose that we have only one polymatroid $h$ on ground set $N$
and fix $L \subseteq N$. An \emph{$L$-copy} of $N$ is a finite set $M$
with $|M| = |N|$ and $M \cap N = L$. We fix a bijection $\pi: N \to M$
which preserves $L$ pointwise. The polymatroid $h$ is a
\emph{selfadhesive polymatroid at $L$} if there exists a polymatroid
$\ol{h}$ which is an adhesive extension of $h$ and its induced copy $\pi(h)$
over their common restriction to~$L$. The polymatroid is \emph{selfadhesive}
if it is selfadhesive at every $L \subseteq N$. The fundamental result of
\cite{MatusAdhesive} is:

\begin{theorem*}
Any two of the restrictions of an entropic polymatroid have an entropic
adhesive extension. In particular, entropy vectors are selfadhesive.
\end{theorem*}

\begin{remark}
The set of polymatroids on $N$ which are selfadhesive forms a rational,
polyhedral cone in $\BB R^{2^N}$. To see this, let $N$, a subset $L \subseteq N$
and an $L$-copy $M$ of $N$ with bijection $\pi$ be fixed. The conditions
for a pair $(h, \ol{h})$, where $h: 2^N \to \BB R$ and $\ol{h}: 2^{NM}
\to \BB R$, to be polymatroids and $\ol{h}$ to be an adhesive extension
of $h$ and $\pi(h)$ are homogeneous linear equalities and inequalities with
integer coefficients in the entries of $h$ and $\ol{h}$. Hence, the set of
such pairs is a rational, polyhedral cone in $\BB R^{2^N} \times \BB R^{2^{NM}}$.
By the Fourier--Motzkin elimination theorem \cite[Theorem~1.4]{Ziegler}, these
properties are inherited by the projection down to $\BB R^{2^N}$ which consists
of all polymatroids $h$ which are selfadhesive at~$L$. Intersecting these
cones for all $L$ gives the desired set of selfadhesive polymatroids and shows
that this set is a rational, polyhedral cone.
\end{remark}

\begin{remark}
Linear inequalities which are valid for entropic polymatroids are called
\emph{information inequalities}. The above observation implies that selfadhesivity,
as a necessary condition for entropicness, captures only finitely many
information inequalities for each fixed~$N$. By contrast, Matúš~\cite{MatusInfinite}
showed that even for $|N| = 4$ there are infinitely many irredundant information
inequalities.

In the $|N| = 4$ case, the cone of selfadhesive polymatroids is characterized
(in addition to the polymatroid properties) by the validity of the Zhang--Yeung
inequalities (see \Cref{rem:ZY}). In~this sense, selfadhesivity is a reformulation
of the Zhang--Yeung inequalities using only the notions of restriction and
conditional independence. The generalization of the concept of adhesive extension
to more than one $L$-copy of a polymatroid leads to the \emph{book inequalities}
of~\cite{BookIneq}.
\end{remark}

\section{Adhesive extensions of Gaussians}

The analogous result for Gaussian covariance matrices is our main theorem:

\begin{theorem} \label{thm:Adhe}
Let $\Sigma \in \PD_N$ and $\Sigma' \in \PD_M$ be two covariance
matrices with common restriction $\Sigma_L = \Sigma'_L$, where
$L = N \cap M$. There exists a unique $\Phi \in \PD_{NM}$
such that:
\begin{itemize}[itemsep=0pt]
\item $\Phi_N = \Sigma$ and $\Phi_M = \Sigma'$,
\item $\CI{N,M|L} \in \CIS{\Phi}$.
\end{itemize}
\end{theorem}

\begin{proof}
Let $N' = N \setminus L$, $M' = M \setminus L$. We use the following
names for blocks of $\Sigma$ and~$\Sigma'$:
\[
  \Sigma = \kbordermatrix{
       & L    & N' \\
    L  & X    & A  \\
    N' & A^\T & Y
  }, \qquad\qquad
  \Sigma' = \kbordermatrix{
       & L    & M' \\
    L  & X    & B \\
    M' & B^\T & Z
  }.
\]
Consider the matrix
\[
  \Phi = \kbordermatrix{
       & L    & N'         & M'      \\
    L  & X    & A          & B       \\
    N' & A^\T & Y          & \Lambda \\
    M' & B^\T & \Lambda^\T & Z
  },
\]
where $\Lambda$ will be determined shortly. Its restrictions to $N$ and
$M$ are clearly equal to~$\Sigma$ and $\Sigma'$, respectively.
The CI~statement $\CI{N,M|L}$ is equivalent to the rank requirement
$\Rk \Phi_{N,M} = |N \cap M| = |L|$, but then rank additivity of the
Schur complement shows
\[
  |L| = \Rk \Phi_{N,M} = \Rk \begin{pmatrix} X & B \\ A^\T & \Lambda\end{pmatrix} =
  \underbrace{\Rk X}_{= |L|} + \Rk(\Lambda - A^\T X^{-1} B).
\]
This implies $\Lambda = A^\T X^{-1} B$ and thus $\Phi$ is uniquely
determined by $\Sigma$ and $\Sigma'$ via the two conditions in the
\namecref{thm:Adhe}.
To show positive definiteness, consider the transformation
\[
  P = \kbordermatrix{
       & L   & N'        & M'        \\
    L  & \Id & -X^{-1} A & -X^{-1} B \\
    N' & 0   & \Id       & 0         \\
    M' & 0   & 0         & \Id
  }
\]
of the bilinear form $\Phi$:
\[
  P^\T \Phi P = \begin{pmatrix}
    X & 0 & 0 \\
    0 & Y - A^\T X^{-1} A & 0 \\
    0 & 0 & Z - B^\T X^{-1} B
  \end{pmatrix} =
  \begin{pmatrix}
    \Sigma_L & 0 & 0 \\
    0 & \Sigma^L & 0 \\
    0 & 0 & \Sigma'^L
  \end{pmatrix}.
\]
The result is clearly positive~definite and since $P$ is invertible,
this shows $\Phi \in \PD_{NM}$.
\end{proof}

\begin{remark} \label{rem:Machinery}
An alternative proof of this theorem was kindly pointed out by one of the
referees. It relies on viewing the existence of $\Phi$ as a positive definite
matrix completion problem where the entries of $\Phi_N$ and $\Phi_M$ are
prescribed and the submatrix $\Phi_{N',M'}$ is left unspecified. The machinery
developed in \cite{PosdefCompletion} shows that a positive definite completion
exists and that there is a unique completion $\Psi$ with maximal determinant.
This matrix satisfies $(\Psi^{-1})_{N',M'} = 0$ which is equivalent to
$\CI{N,M|L}$ by the duality concept in Gaussian CI~theory; cf.~\cite[Proposition~3.10]{Dissert}.
\end{remark}

\begin{remark} \label{rem:ZY}
Zhang and Yeung~\cite{ZhangYeung} proved the first information inequality for
entropy vectors which is not a consequence of the Shannon inequalities
(equivalently, the polymatroid properties). It can be expressed as
the non-negativity of the functional
\[
  \CIz{i,j|kl} \defas \CId{kl|i} + \CId{kl|j} + \CId{ij|} - \CId{kl|}
    + \CId{ik|l} + \CId{il|k} + \CId{kl|i}.
\]
Matúš~\cite{MatusAdhesive} characterized the selfadhesive polymatroids over
a 4-element ground set as those polymatroids satisfying $\CIz{i,j|kl} \ge 0$
for all choices of $i,j,k,l$. As a corollary to \Cref{thm:Adhe} we obtain
that the multiinformation vectors and hence the differential entropy vectors
of Gaussian distributions satisfy the Zhang--Yeung inequalities. This is one
half of the result proved by Lněnička~\cite{Lnenicka}. However, that result
also follows from the metatheorem of Chan~\cite{ChanBalanced} since $\CIz{i,j|kl}$
is balanced.
\end{remark}

In the theory of regular Gaussian conditional independence structures,
it is natural to relax the positive definiteness assumption on $\Sigma$
to that of \emph{principal regularity}, i.e., all principal minors,
instead of being positive, are required not to vanish. Principal
regularity is the minimal technical condition which allows the formation
of all Schur complements and the property is inherited by principal
submatrices and Schur complements, hence enabling analogues of
marginalization and conditioning over general fields instead of
the field $\BB R$; see \cite{Gaussant} for applications. However,
the last step in the above proof of \Cref{thm:Adhe} requires
positive definiteness and does not work for principally regular matrices:

\begin{example} \label{ex:PRnotSelfadhe}
Consider the following principally regular matrix over $N = ijkl$:
\[
  \renewcommand{\arraystretch}{1.2}
  \Gamma = \kbordermatrix{
      & i               &  j                &        & k                        & l                 \\
    i & 1               &  0                & \vrule & 0                        & \sfrac1{\sqrt2}   \\
    j & 0               &  1                & \vrule & \sfrac1{2\sqrt2}         & 0                 \\ \cline{2-6}
    k & 0               &  \sfrac1{2\sqrt2} & \vrule & 1                        & \sfrac{\sqrt3}{2} \\
    l & \sfrac1{\sqrt2} &  0                & \vrule & \sfrac{\sqrt3}{2}        & 1
  }
\]
and fix $L = ij$. By the proof of \Cref{thm:Adhe}, the submatrix and rank
conditions uniquely determine an adhesive extension of $\Gamma$ with an
$L$-copy of itself over the ground set $\I{ijklk'l'}$. This unique candidate
matrix~is
\[
  \renewcommand{\arraystretch}{1.2}
  \kbordermatrix{
       & i               &  j                &        & k                        & l                 &        & k'                       & l'                \\
    i  & 1               &  0                & \vrule & 0                        & \sfrac1{\sqrt2}   & \vrule & 0                        & \sfrac1{\sqrt2}   \\
    j  & 0               &  1                & \vrule & \sfrac1{2\sqrt2}         & 0                 & \vrule & \sfrac1{2\sqrt2}         & 0                 \\ \cline{2-9}
    k  & 0               &  \sfrac1{2\sqrt2} & \vrule & 1                        & \sfrac{\sqrt3}{2} & \vrule & \sfrac18                 & 0                 \\
    l  & \sfrac1{\sqrt2} &  0                & \vrule & \sfrac{\sqrt3}{2}        & 1                 & \vrule & 0                        & \sfrac12          \\ \cline{2-9}
    k' & 0               &  \sfrac1{2\sqrt2} & \vrule & \sfrac18                 & 0                 & \vrule & 1                        & \sfrac{\sqrt3}{2} \\
    l' & \sfrac1{\sqrt2} &  0                & \vrule & 0                        & \sfrac12          & \vrule & \sfrac{\sqrt3}{2}        & 1
  }.
\]
But this matrix is not principally regular, as the $\I{lk'l'}$-principal
minor is zero. However, the CI~structure $\CC G = \CIS{\Gamma}$
is the dual of the graphical model for the undirected path {$i$~--~$l$~--~$k$~--~$j$};
cf.~\cite[Section~3]{LnenickaMatus}.
This implies that $\CC G$ is representable by a positive definite
matrix with rational entries and even though the particular matrix representation
$\Gamma$ does not have a selfadhesive extension (in the sense of \Cref{thm:Adhe}),
another representation of $\CC G$ exists which is positive definite and hence
selfadhesive.
\end{example}

\section{Structural selfadhesivity}

The existence of adhesive extensions and in particular selfadhesivity
of positive~definite matrices induces similar properties on their
CI~structures, since the conditions in \Cref{thm:Adhe} can be formulated
using only the concepts of restriction and conditional independence.
On the CI~level, we sometimes use the term \emph{structural selfadhesivity}
to emphasize that it is generally a weaker notion than what is proved for
covariance matrices above. Selfadhesivity can be used to strengthen
known properties of CI~structures: if it is known that all positive~definite
matrices have a certain distinguished property $\FR p$, then the fact that
$\Sigma$ and any $L$-copy of it fit into an adhesive, positive~definite
extension obeying $\FR p$ says more about the structure of~$\Sigma$ than
$\FR p$ alone. We begin by making precise the notion of a \emph{property}:

\begin{definition}
Let $\FR A_N = 2^{\A_N}$ be the set of all CI~structures over $N$.
For $N = [n] = \Set{1, \dots, n}$ we use abbreviations $\A_n$ and
$\FR A_n$.
A \emph{property} of CI~structures is an element $\FR p$ of the
\emph{property lattice}
\[
  \FR P \defas \bigtimes_{n=1}^\infty 2^{\FR A_n}.
\]
\end{definition}

A property $\FR p$ consists of one set $\FR p(n) \subseteq \FR A_n$ per
finite cardinality~$n$. This is the set of CI~structures over~$[n]$
which ``have property $\FR p$''. CI~structures $\CC L$ and $\CC M$ over $N$
and $M$, respectively, are \emph{isomorphic} if there is a bijection
$\pi: N \to M$ such that under the induced map $\CC M = \pi(\CC L)$.
We are only interested in properties which are invariant under isomorphy.
Hence, the choice of ground sets $[n]$ presents no restriction. Moreover,
we freely identify isomorphic CI~structures in the following. In particular,
each $k$-element subset $K \subseteq [n]$ will be tacitly identified
with~$[k]$ and we use notation such as $\FR p(K)$.

\begin{example}
By the localization rule \eqref{eq:L'}, the well-known semigraphoid axioms
of \cite{PearlPazGraphoids} reduce to the single inference rule
\[
  \label{eq:sg} \tag{S}
  \CI{i,j|L} \wedge \CI{i,k|jL} \Rightarrow \CI{i,j|kL} \wedge \CI{i,k|L}.
\]
Being a semigraphoid is a property defined by
\begin{gather*}
  \FR{sg}(n) \defas \Set{ \CC L \subseteq \A_n : \text{\eqref{eq:sg} holds for
  $\CC L$ for all $ijk \in \binom{[n]}{3}$ and $L \subseteq [n] \setminus ijk$} }.
\end{gather*}
Being realizable by a Gaussian distribution is another property
\[
  \FR g^+(n) \defas \Set{ \CIS{\Sigma} \in \A_n : \Sigma \in \PD_n }.
\]
Both are closed under restriction, which can be expressed as follows:
for every $\CC L \in \FR p(N)$ and every $K \subseteq N$ we have
$\CC L|_K \defas \CC L \cap \A_K \in \FR p(K)$.
\end{example}

The property lattice is equipped with a natural order relation of
component-wise set inclusion from the boolean lattices $2^{\FR A_n}$.
This order relation $\le$ compares properties by generality:
if $\FR p \le \FR q$, then for all $n \ge 1$ we have $\FR p(n)
\subseteq \FR q(n)$, and $\FR p$ is \emph{sufficient} for $\FR q$
and, equivalently, $\FR q$ is \emph{necessary} for~$\FR p$.
A function $\varphi$ on the property lattice is \emph{recessive}
if for every $\FR p \in \FR P$ we have $\varphi(\FR p) \le \FR p$.
It is \emph{monotone} if $\FR p \le \FR q$ entails $\varphi(\FR p)
\le \varphi(\FR q)$.

\begin{definition} \label{def:selfadhe}
Let $\FR p$ be a property of CI~structures. The \emph{selfadhesion}
$\FR p^\sa(N)$ of $\FR p$ is the set of CI~structures $\CC L$ such that
for every $L \subseteq N$ together with an $L$-copy $M$ of $N$ and
bijection $\pi: N \to M$ there exists $\ol{\CC L} \in \FR p(NM)$
satisfying the conditions:
\begin{itemize}[itemsep=0pt]
\item $\ol{\CC L}|_N = \CC L$, $\ol{\CC L}|_M = \pi(\CC L)$, and
\item $\CI{N,M|L} \in \ol{\CC L}$.
\end{itemize}
A~property is \emph{selfadhesive} if $\FR p = \FR p^\sa$.
\end{definition}

The following is a direct consequence of \Cref{thm:Adhe}:

\begin{corollary} \label{cor:Gaussadhe}
The property $\FR g^+$ of being regular Gaussian is selfadhesive.
\end{corollary}

\begin{proof}
Let $\CC L \in \FR g^+(N)$ be Gaussian and $\Sigma \in \PD_N$ a realizing
matrix. For any $L \subseteq N$, \Cref{thm:Adhe} applies with $\Sigma' =
\Sigma$ and gives a matrix $\Phi$ whose CI~structure is a witness for the
structural selfadhesivity of $\CC L$ at~$L$.
\end{proof}

\begin{lemma} \label{lemma:AdheProp}
The operator $\cdot^\sa$ is recessive and monotone on the property lattice.
\end{lemma}

\begin{proof}
Let $\FR p$ be a property and $\CC L \in \FR p^\sa(N)$.
In particular, $\CC L$ is selfadhesive with respect to $\FR p$ at $L = N$.
The $L$-copy $M$ of $N$ in the definition must be $M = N$ and it follows
that $\CC L \in \FR p(NM) = \FR p(N)$. This proves recessiveness
$\FR p^\sa \le \FR p$.
For monotonicity, let $\FR p \le \FR q$ and $\CC L$ in $\FR p^\sa(N)$.
Then for every $L$ with $L$-copy $M$ of $N$ there exists a certificate
for the existence of $\CC L$ in $\FR p^\sa$. This certificate lives in
$\FR p(NM) \subseteq \FR q(NM)$ which proves $\CC L \in \FR q^\sa(N)$.
\end{proof}

Thus, from monotonicity and the fact that $\FR g^+$ is a fixed point of
self\-adhesion, we can conclude that a property which is necessary
for Gaussianity remains necessary after selfadhesion. Since selfadhesion
makes properties more specific, this allows us to take known necessary
properties of Gaussian~CI and to derive new, stronger properties from them.

\begin{corollary} \label{cor:AdheNecessary}
If $\FR g^+ \le \FR p$, then $\FR g^+ \le \FR p^\sa$.
\qed
\end{corollary}

Iterated application of selfadhesion gives rise to a chain of ever more
specific properties $\FR g^+ \le \cdots \le \FR p^{k \cdot \sa} \le \cdots
\le \FR p^{2\cdot \sa} \defas (\FR p^\sa)^\sa \le \FR p^\sa \le \FR p$.
For~each fixed component $n$ of the property, this results in a descending
chain in the finite boolean lattice $2^{\FR A_n}$ which must stabilize
eventually. However, the whole property $\FR p$ has a countably infinite
number of components and it is not clear if iterated selfadhesions
converge after finitely many steps to the limit $\FR p^{\omega\cdot \sa}
\defas \bigwedge_{k=1}^\infty {\FR p^{k\cdot \sa}}$ in the property lattice.

\begin{question}
Does $\cdot^\sa$ stabilize after the first application to ``well-behaved''
properties like~$\FR{sg}$, i.e., is $\FR{sg}^\sa = \FR{sg}^{\omega\cdot\sa}$?
Under which assumptions on a property does $\cdot^\sa$ stabilize after a
finite number of applications?
\end{question}

We now turn to the question which closure properties of $\FR p$ are recovered
for $\FR p^\sa$. For example, if for every $\CC L, \CC L' \in \FR p(N)$
we have $\CC L \cap \CC L' \in \FR p(N)$, then $\FR p$ is \emph{closed
under intersection}. Semigraphoids enjoy this closure property because
they are axiomatized by the Horn clauses~\eqref{eq:sg}. The following
\namecref{lemma:AdheIntersect} shows that all iterated selfadhesions
inherit closure under intersection.

\begin{lemma} \label{lemma:AdheIntersect}
If $\FR p$ is closed under intersection, then so is $\FR p^\sa$.
\end{lemma}

\begin{proof}
Let $\CC L, \CC L' \in \FR p^\sa(N)$ and fix a set $L \subseteq N$ and an
$L$-copy $M$ of $N$ with bijection~$\pi$. There are $\ol{\CC L}$ and
$\ol{\CC L'}$ in $\FR p(NM)$ witnessing the selfadhesivity of $\CC L$ and
$\CC L'$, respectively, at~$L$. Their intersection $\ol{\CC L} \cap \ol{\CC L'}$
is in $\FR p(NM)$ by assumption and we have
\begin{itemize}[itemsep=0pt, wide]
\item $(\ol{\CC L} \cap \ol{\CC L'})|_N
  = \ol{\CC L}|_N \cap \ol{\CC L'}|_N
  = \CC L \cap \CC L'$,
\item $(\ol{\CC L} \cap \ol{\CC L'})|_M
  = \ol{\CC L}|_M \cap \ol{\CC L'}|_M
  = \pi(\CC L) \cap \pi(\CC L') = \pi(\CC L \cap \CC L')$,
\item $\CI{N,M|L} \in \ol{\CC L} \cap \ol{\CC L'}$.
\end{itemize}
Thus it proves selfadhesivity of $\CC L \cap \CC L'$ with respect to
$\FR p$ at~$L$.
\end{proof}

Similarly to matroid theory, \emph{minors} are the natural subconfigurations
of CI~structures. They are the CI-theoretic abstraction of marginalization
and conditioning on random vectors.

\begin{definition}
Let $\CC L \subseteq \CC A_N$ and $x \in N$. The \emph{marginal} and
the \emph{conditional} of $\CC L$ on $N \setminus x$ are, respectively,
\begin{align*}
  \CC L \mathbin{\backslash} x &\defas \Set{ \CI{i,j|K} \in \CC A_{N \setminus x} : \CI{i,j|K} \in \CC L } = \CC L \cap \CC A_{N \setminus x}, \\
  \CC L \mathbin{/}          x &\defas \Set{ \CI{i,j|K} \in \CC A_{N \setminus x} : \CI{i,j|xK} \in \CC L}.
\end{align*}
A \emph{minor} of $\CC L$ is any CI~structure which is obtained by a
sequence of marginalizations and conditionings.
\end{definition}

If for every $\CC L \in \FR p(N)$ and every minor $\CC K$ of $\CC L$
on ground set $M \subseteq N$ we have $\CC K \in \FR p(M)$, then $\FR p$
is \emph{minor-closed}. Minor-closedness is necessary for the existence
of a finite axiomatization of a property~$\FR p$. More concretely,
\cite{MatusMinors} studied descriptions of properties by finitely
many ``forbidden minors'', which is under natural regularity assumptions
equivalent to having a finite axiomatic description by boolean
CI~inference formulas; cf.~\cite[Section~4.4]{Dissert} for details.

\begin{lemma}
If $\FR p \le \FR{sg}$ is minor-closed, then so is $\FR p^\sa$.
\end{lemma}

\begin{proof}
By induction it suffices to prove closedness under marginals and conditionals.
Let $\CC L \in \FR p^\sa(N)$ and $x \in N$. First, we prove that $\CC L
\mathbin{\backslash} x \in \FR p^\sa(N \setminus x)$. Fix $L \subseteq N \setminus x$
and an $L$-copy $M$ of $N$ with bijection~$\pi$ and let $\ol{\CC L}$ be
the witness for selfadhesivity of $\CC L$ at~$L$. The minor $\ol{\CC L}
\mathbin{\backslash} \Set{x, \pi(x)}$ is in $\FR p(NM \setminus \Set{x, \pi(x)})$
by assumption of minor-closedness; and note that $M \setminus \pi(x)$ is
an $L$-copy of $N \setminus x$. Moreover, $(\ol{\CC L} \mathbin{\backslash}
\Set{x, \pi(x)})|_{N \setminus x} = \CC L \mathbin{\backslash} x$ which is
isomorphic to $\pi(\CC L \mathbin{\backslash} x) = \pi(\CC L)
\mathbin{\backslash} \pi(x) = (\ol{\CC L} \mathbin{\backslash}
\Set{x, \pi(x)})|_{M \setminus \pi(x)}$.
For the last argument we need the semigraphoid property to hold for~$\FR p$.
This ensures by \cite[Lemma~2.2]{Studeny} that the localization rule~\eqref{eq:L'}
applies. This rule shows that $\CI{N,M|L} \in \ol{\CC L}$ is equivalent to
\[
  \bigwedge_{\substack{i \in N', j \in M', \\
    L \subseteq P \subseteq NM \setminus ij}} \CI{i,j|P} \in \ol{\CC L}.
\]
Applying the rule \eqref{eq:L'} again in reverse to a subset of these
elementary CI~statements shows that $\CI{(N\setminus x), (M \setminus \pi(x))|L} \in
\ol{\CC L}$ holds, which finishes the proof that $\ol{\CC L} \mathbin{\backslash}
\Set{x, \pi(x)}$ is a witness for the selfadhesion of $\CC L \mathbin{\backslash} x$
at $L$.

To prove that $\CC L \mathbin{/} x \in \FR p^\sa(N \setminus x)$, pick
any $L \subseteq N \setminus x$ and let $M$ be an $Lx$-copy of $N$ with
bijection~$\pi$. Note that $M \setminus x$ is an $L$-copy of $N \setminus x$
with bijection $\pi|_{N \setminus x}$. Let $\ol{\CC L} \in \FR p(NM)$ be a
witness for the selfadhesivity of $\CC L$ at~$Lx$ and consider the conditional
$\ol{\CC L}
\mathbin{/} x$:
\begin{align*}
     (\ol{\CC L} \mathbin{/} x)|_{N \setminus x}
  &= \Set{ (ij|K) \in \CC A_{N \setminus x} : (ij|Kx) \in \ol{\CC L} } \\
  &= (\ol{\CC L}|_N) \mathbin{/} x = \CC L \mathbin{/} x.
\end{align*}
An analogous computation shows $(\ol{\CC L} \mathbin{/} x)|_{M \setminus x}
= \pi(\CC L \mathbin{/} x)$ using that $x$ is fixed by~$\pi$. Moreover, we
have $\CI{N,M|Lx} \in \ol{\CC L}$ which is equivalent to $\CI{(N\setminus x),
(M\setminus x)|Lx} \in \ol{\CC L}$ since $x \in N \cap M$. But this entails
$\CI{(N\setminus x),(M\setminus x)|L} \in \ol{\CC L} \mathbin{/} x$ and hence
$\CC L \mathbin{/} x$ is selfadhesive at $L$ with witness~$\ol{\CC L}
\mathbin{/} x$.
\end{proof}

\begin{question}
Does $\FR{sg}^\sa$ have a finite axiomatization? Is finite axiomatizability
or finite non-axiomatizability in general preserved by selfadhesion?
\end{question}

\subsection{Selfadhesivity testing}

Whether or not a CI~structure $\CC L \subseteq \A_N$ is in $\FR p^\sa(N)$
can be checked algorithmically if an \emph{oracle} $\SF p(\tilde{\CC L})$
for the property~$\FR p$ is available. This oracle is a subroutine which
receives a \emph{partially defined} CI~structure $\tilde{\CC L}$ over~$N$,
i.e., a set of CI~statements or negated CI~statements specifying constraints
on some statements from~$\A_N$. Then~$\SF p$ decides if $\tilde{\CC L}$ can
be extended to a member of~$\FR p(N)$.

\begin{algorithm}[h!]
\caption{Blackbox selfadhesion membership test}
\algblockdefx[function]{BeginFunction}{EndFunction}{\textbf{function}\xspace}{\textbf{end function}}
\begin{algorithmic}[1]
\BeginFunction $\textsf{is-selfadhesive}(\CC L, \SF p)$ \Comment{tests if $\CC L \in \FR p^\sa(N)$}
\ForAll{$L \subseteq N$}
  \State $(M, \pi) \gets \text{$L$-copy of $N$ with bijection $\pi: N \to M$}$
  \State $\tilde{\CC L} \gets \emptyset$
  \ForAll{$s \in \A_N$}
    \State \textbf{if} $s \in \CC L$     \textbf{then} $\tilde{\CC L} \gets \tilde{\CC L} \cup \Set{ \hphantom\neg s, \hphantom\neg \pi(s) }$
    \State \textbf{if} $s \not\in \CC L$ \textbf{then} $\tilde{\CC L} \gets \tilde{\CC L} \cup \Set{          \neg s,          \neg \pi(s) }$
  \EndFor
  \State $\tilde{\CC L} \gets \tilde{\CC L} \cup \Set{ \CI{N,M|L} }$ \Comment{or equivalent statements via \eqref{eq:L'}}
  \State \textbf{if} $\SF p(\tilde{\CC L}) = \TT{false}$ \textbf{then} \textbf{return} $\TT{false}$
\EndFor
\State \textbf{return} $\TT{true}$
\EndFunction
\end{algorithmic}
\label{alg:Selfadhe}
\end{algorithm}

Each component $\FR p(n)$ of a property $\FR p$ is a set of subsets of~$\CC A_n$.
There are two principal ways of representing this set: \emph{explicitly}, by
listing its elements, or \emph{implicitly}, by listing a set of abstract axioms
in the form of boolean formulas which all its elements and no other CI~structures
satisfy. A typical application of \Cref{alg:Selfadhe} takes in both, an explicit
description of $\FR p(n)$ to iterate over, as well as an implicit description~$\SF p$
of~$\FR p$ to perform selfadhesivity testing for ground sets of sizes between $n$
and~$2n$. It~outputs only an explicit description of $\FR p^\sa$ at a given index~$n$.
Transforming this explicit description obtained from \Cref{alg:Selfadhe}
into an implicit description to call the algorithm again is akin to transforming
a disjunctive normal form of a boolean formula into a conjunctive normal form, which
is a hard problem. Moreover, it would be required to compute $\FR p^\sa(m)$
explicitly for all $n \le m \le 2n$. This makes it difficult to iterate selfadhesions.

\begin{remark}
The proof of \Cref{lemma:AdheProp} shows that a CI~structure $\CC L$ satisfies
selfadhesivity with respect to $\FR p$ at $L = N$ if and only if $\CC L$ has
property $\FR p$. In the other extreme case, every structure in $\FR p$ is
selfadhesive at $L = \emptyset$ if $\FR p$ is closed under the direct sum
operation introduced in \cite{MatusMatroids}.
Many useful properties are closed under direct sums because this operation
mimics the independent joining of two random vectors; see \cite{MatusClassification}.
If this is known a~priori, some selfadhesivity tests can be skipped.
\end{remark}

We now proceed to apply \Cref{alg:Selfadhe} to two practically tractable
necessary conditions for Gaussian realizability. The computational results
allow, via \Cref{cor:AdheNecessary}, the deduction of new CI~inference
axioms for Gaussians on five random variables.

\subsection{Structural semigraphoids}

It is easy to see that every Gaussian CI~structure $\CC L = \CIS{\Sigma}$
can also be obtained from the \emph{correlation matrix} $\Sigma'$ of the
original distribution $\Sigma$. Hence, we may assume that $\Sigma$ is a
correlation matrix. In that case, the \emph{multiinformation vector} of
$\Sigma$ is the map $m_\Sigma: 2^N \to \BB R$ given by $m_\Sigma(K)
\defas -\sfrac12 \log \det \Sigma_K$.
This function satisfies $m_\Sigma(\emptyset) = m_\Sigma(i) = 0$ for all
$i \in N$ and it is super\-modular by the Koteljanskii inequality; see \cite{JohnsonBarrett}.
Similarly to entropy vectors, the equality condition in these inequalities
characterizes conditional independence: ${\CId{ij|K} \cdot m_\Sigma = 0}
\;\Leftrightarrow\; {\CI{i,j|K} \in \CIS{\Sigma}}$.

In the nomenclature of \cite[Chapter~5]{Studeny}, $m_\Sigma$ is an
\emph{$\ell$-standardized super\-modular function}. The functions having
these two properties form a rational, polyhedral cone $\BO{S}_N$ of
codimension $|N| + 1$ in $\BB R^{2^N}$. Each of its facets is
given by equality in precisely one of the supermodular inequalities
$\CId{ij|K} \le 0$ for an elementary CI~statement $\CI{i,j|K} \in \A_N$.
Since the facets of this cone are in bijection with CI~statements,
it is natural to identify faces (intersections of facets) dually with
CI~structures (unions of CI~statements). The property of CI~structures
defined by arising from a face of $\BO{S}_N$ is that of \emph{structural
semigraphoids}, denoted by $\FR{sg}_*$, and it is necessary for $\FR g^+$
since every Gaussian CI~structure $\CIS{\Sigma}$ is associated with the
unique face on which  $m_\Sigma \in \BO{S}_N$ lies in the relative interior.

\begin{remark}
Structural semigraphoids can be equivalently defined via the face
lattice of the cone of \emph{tight} polymatroids, i.e., polymatroids~$h$
with $h(N) = h(N \setminus i)$ for every $i \in N$. The tightness
condition poses no extra restrictions: for every polymatroid, there
exists a tight polymatroid inducing the same pure CI~statements
(only differing in the functional dependences); cf.~\cite[Section~III]{MatusCsirmaz}.
A~proof of the equivalence is contained in \cite[Section~6.3]{Dissert},
\end{remark}

Deciding whether a partially defined CI~structure $\tilde{\CC L}$
is consistent with the structural semigraphoid property is a question
about the incidence structure of the face lattice of~$\BO{S}_N$.
Such questions reduce to the feasibility of a rational linear program
as previously demonstrated by \cite{EfficientCI}. \Cref{alg:Struct}
relies on this insight by setting up the polyhedral description of the
structural semigraphoidality test and then delegating the computation
to specialized linear programming software.

\begin{algorithm}[h!]
\caption{Structural semigraphoid consistency test}
\algblockdefx[function]{BeginFunction}{EndFunction}{\textbf{function}\xspace}{\textbf{end function}}
\begin{algorithmic}[1]
\BeginFunction $\textsf{is-structural}(\tilde{\CC L})$ \Comment{tests if $\tilde{\CC L}$ is consistent with $\FR{sg}_*(N)$}
\State $P \gets \Set{ \text{$m(\emptyset) = m(i) = 0$ for all $i \in N$} }$ \Comment{$H$ description of polyhedron}
\ForAll{$s \in \A_N$}
  \State \textbf{if} $\hphantom\neg s \in \tilde{\CC L}$ \textbf{then} $P \gets P \cup \Set{ -\CId{s} \cdot m  =  0 }$
  \State \textbf{if} $         \neg s \in \tilde{\CC L}$ \textbf{then} $P \gets P \cup \Set{ -\CId{s} \cdot m \ge 1 }$
  \State \Comment{The condition $-\CId{s} \cdot m > 0$ is equivalent to $\ge 1$ in a cone}
  \State \textbf{else} $P \gets P \cup \Set{ -\CId{s} \cdot m \ge 0 }$
\EndFor
\State \textbf{return} $\textsf{is-feasible}(P)$ \Comment{call an \TT{LP} solver}
\EndFunction
\end{algorithmic}
\label{alg:Struct}
\end{algorithm}

Equipped with this oracle for $\FR{sg}_*$, \Cref{alg:Selfadhe} can be applied
to compute membership in~$\FR{sg}_*^\sa$. We run the structural selfadhesivity
test for the \emph{gaussoids} of \cite{LnenickaMatus} because they are easily
computable candidates for Gaussian CI~structures; see also \cite{Geometry}.
For~$n=4$ random variables, the gaussoids which are structural semigraphoids
already coincide with the realizable Gaussian structures (as classified in
\cite{LnenickaMatus}) and selfadhesivity offers no improvement.
This is no longer the case for five random variables:

\begin{computation}
There are $508\,817$ gaussoids on $n=5$ random variables modulo isomorphy.
Of~these $336\,838$ are structural semigraphoids and $335\,047$ of them are
selfadhesive with respect to~$\FR{sg}_*$.
\end{computation}

A semigraphoid $\CC L$ is structural if and only if it is induced by a
polymatroid, i.e., $\CC L = \CIS{h}$. In~this case, two distinct notions of
selfadhesivity can be applied to $\CC L$: the first is Matúš's definition
of selfadhesivity for the inducing polymatroid~$h$; and the second is
structural selfadhesivity from \Cref{def:selfadhe} for the CI~structure
$\CC L$ with respect to the property~$\FR{sg}_*$. Analogously to
\Cref{cor:Gaussadhe}, one sees that the second condition is implied by the first.
The existence of a selfadhesive inducing polymatroid can be efficiently
tested for ground set size four based on the polyhedral description of
the cone of selfadhesive 4-polymatroids from \cite[Corollary~6]{MatusAdhesive}.

\begin{computation}
Out of the $1\,285$ isomorphy representatives of $\FR{sg}_*(4)$,
exactly $1\,224$ are in $\FR{sg}_*^\sa(4)$. Each of them is induced by
a selfadhesive 4-polymatroid.
\end{computation}

\begin{question}
Is every element of $\FR{sg}_*^\sa(N)$ induced by a selfadhesive
$N$-polymatroid, for every finite set~$N$?
\end{question}

\subsection{Orientable gaussoids}

Recall from \cite{Geometry} that a gaussoid is \emph{orientable} if it is
the support of an oriented gaussoid. Oriented gaussoids are a variant of
CI~structures in which every statement $\CI{i,j|K}$ has a sign $\Set{ \TT0,
\TT+, \TT- }$ attached, indicating conditional independence, positive or
negative partial correlation, respectively. Oriented gaussoids are
axiomatically defined and therefore \TT{SAT} solvers are ideally suited
to decide the consistency of a partially defined CI~structure with these
axioms. The property of orientability, denoted $\FR{o}$, is obtained from
the set of oriented gaussoids by mapping all CI~statements oriented as
$\TT0$ to elements of a CI~structure and all statements oriented $\TT+$
or $\TT-$ to non-elements. To~facilitate orientability testing, one
allocates two boolean variables $V^{\TT0}_s$ and $V^{\TT+}_s$ for every
CI~statement~$s$. The~former indicates whether $s$ is $\TT0$ or not while
the latter indicates, provided that $V^{\TT0}_s$ is false, if $s$ is
$\TT+$ or $\TT-$.
Further details about oriented gaussoids, their axioms and use of \TT{SAT}
solvers for CI~inference are available in \cite{Geometry}.
\Cref{alg:Orient} gives a condensed account of the algorithm.

\begin{algorithm}[h!]
\caption{Orientable gaussoid consistency test}
\algblockdefx[function]{BeginFunction}{EndFunction}{\textbf{function}\xspace}{\textbf{end function}}
\begin{algorithmic}[1]
\BeginFunction $\textsf{is-orientable}(\tilde{\CC L})$ \Comment{tests if $\tilde{\CC L}$ is consistent with $\FR o(N)$}
\State $\varphi \gets \textsf{oriented-gaussoid-axioms}(N)$ \Comment{boolean formula}
\ForAll{$s \in \A_N$}
  \State \textbf{if} $\hphantom\neg s \in \tilde{\CC L}$ \textbf{then} $\varphi \gets \varphi \wedge [V^{\TT0}_s = \TT{true}]$
  \State \textbf{if} $         \neg s \in \tilde{\CC L}$ \textbf{then} $\varphi \gets \varphi \wedge [V^{\TT0}_s = \TT{false}]$
  \State $\varphi \gets \varphi \wedge [V^{\TT0}_s = \TT{true} \Rightarrow V^{\TT+}_s = \TT{false}]$
  \State \Comment{there are only three signs $\Set{\TT0, \TT+, \TT-}$}
\EndFor
\State \textbf{return} $\textsf{is-satisfiable}(\varphi)$ \Comment{call a \TT{SAT} solver}
\EndFunction
\end{algorithmic}
\label{alg:Orient}
\end{algorithm}

\begin{computation}
All orientable gaussoids on $n=4$ are Gaussian.
Of the $508\,817$ isomorphy classes of gaussoids on $n=5$ precisely
$175\,215$ are orientable and $168\,010$ are selfadhesive with respect
to orientability.
\end{computation}

\subsection{Structural orientable gaussoids}

The meet $\FR{sg}_* \wedge \FR{o}$ of structural semigraphoids
and orientable gaussoids in the property lattice is likewise necessary
for Gaussianity and an oracle for it can be combined from the oracles
of its two constituents. Its selfadhesion yields no improvement over
apparently weaker properties:

\begin{computation}
The properties $\FR{sg}_* \wedge \FR o$ and $\FR{sg}_*^\sa \wedge \FR o$
coincide at $n=5$ with $175\,139$ isomorphy types. On the other hand,
$\FR{sg}_* \wedge \FR{o}^\sa$, $\FR{sg}_*^\sa \wedge \FR{o}^\sa$ and
$(\FR{sg}_* \wedge \FR o)^\sa$ coincide at $n=5$ with $167\,989$~types.
\end{computation}

Up to a few isolated examples in the literature, this represents the
currently best known upper bound in the classification of realizable
Gaussian conditional independence structures on five random variables.
Examination of the difference $(\FR{sg}_* \wedge \FR{o})(5) \setminus
(\FR{sg}_* \wedge \FR{o})^\sa(5)$ reveals new axioms for Gaussian CI
beyond structural semigraphoids and orientability,~e.g.:
\begin{align*}
  \CI{i,j|km} \wedge \CI{i,m|l}  \wedge \CI{j,k|i}  \wedge \CI{j,m}   \wedge \CI{k,l} &\;\Rightarrow\; \CI{i,j}, \\
  \CI{i,k|jl} \wedge \CI{i,l|km} \wedge \CI{j,k|i}  \wedge \CI{j,m|k} \wedge \CI{k,l} &\;\Rightarrow\; \CI{i,k}, \\
  \CI{i,k|j}  \wedge \CI{i,l|jm} \wedge \CI{j,k|il} \wedge \CI{j,m|k} \wedge \CI{k,l} &\;\Rightarrow\; \CI{i,k}.
\end{align*}
The MathRepo page corresponding to this paper contains code and more
information on how to obtain these inference rules algorithmically.
Due to the large amount of data involved and the complexity of minimizing
boolean formulas, it is currently not known how many genuinely new and
mutually irredundant axioms are encoded in the results.

\subsubsection*{Mathematical software and data repository}
\TT{SoPlex v4.0.0} was used to solve rational linear programs exactly;
see \cite{Soplex2012,Soplex2015,SCIP2018}. To check orientability, we
used the incremental \TT{SAT} solver \TT{CaDiCaL v1.3.1} by
\cite{CaDiCaL} and to enumerate satisfying assignments the \TT{AllSAT}
solver \TT{nbc\_minisat\_all v1.0.2} by \cite{TodaSAT}. \Cref{ex:PRnotSelfadhe}
was found using Wolfram \TT{Mathematica v11.3} \cite{Mathematica}.
The~source code and results for all computations are available on the
supplementary MathRepo website of the MPI-MiS and the KEEPER
of the Max-Planck Society:
\begin{center}
  \begin{tabular}{rl}
  MathRepo: & \url{https://mathrepo.mis.mpg.de/SelfadhesiveGaussianCI/} \\
  KEEPER:   & \url{https://keeper.mpdl.mpg.de/d/fbfe463162e94a14ac28/}
  \end{tabular}
\end{center}

\subsubsection*{Acknowledgement}
This project was started at the Otto-von-Guericke-Uni\-ver\-si\-tät
Magdeburg and finished at the Max-Planck Institute for Mathematics in
the Sciences, Leipzig. I~would like to thank the OvGU and the MPI for
providing me with the resources to carry out the computations whose
results are presented here. I also wish to thank the anonymous
referees for their thorough and critical reading of this manuscript,
and especially for drawing my attention to the result mentioned in
\Cref{rem:Machinery}.

\bibliographystyle{tboege}
\bibliography{gaussadhe}

\end{document}